\documentclass[12pt]{article}
\usepackage{amsmath,amssymb,amsthm,subfigure}
\usepackage[dvips]{epsfig}
\usepackage{listings}

\newtheorem{theo}{Theorem}[section]

\newtheorem{cor}[theo]{Corollary}

\begin{document}

\begin{center}
{\large \bf Algorithms for Pattern Containment in $0$-$1$ matrices}
\end{center}

\begin{center}
P.A. CrowdMath
\end{center}

\bigskip

\begin{abstract}
We say a zero-one matrix $A$ avoids another zero-one matrix $P$ if no submatrix of $A$ 
can be transformed to $P$ by changing some ones to zeros. A fundamental problem is to study the extremal function $ex(n,P)$, the maximum number of nonzero entries in 
an $n \times n$ zero-one matrix $A$ which avoids $P$. To calculate exact values of $ex(n,P)$ for specific values of $n$, we need containment algorithms
which tell us whether a given $n \times n$ matrix $A$ contains a given pattern matrix $P$. In this paper, we present optimal algorithms to determine when an $n \times n$ matrix $A$
contains a given pattern $P$ when $P$ is a column of all ones, an identity matrix, a tuple identity matrix, an $L$-shaped pattern, or a cross pattern. These algorithms run in $\Theta(n^2)$ time, 
which is the lowest possible order a containment algorithm can achieve. When $P$ is a rectangular all-ones matrix, we also obtain an improved running time algorithm, albeit with a higher order.
\end{abstract}

\medskip

\section{Introduction}

In this paper, we study matrices or arrays with only two distinct entries, $0$ and $1$,  that avoid certain patterns. We say that a $0$-$1$ matrix $A$ contains another $0$-$1$ matrix $P$ if $A$ has a submatrix that can be transformed into $P$ by changing any number of ones to zeros. Otherwise, $A$ is said to avoid $P$. 

We are interested in algorithms to determine whether an $n \times n$ input matrix $A$ avoids a fixed pattern $P$. Algorithms for pattern containment are naturally related to the classical matrix extremal problem, which seeks to find maximum number of nonzero entries in an $n \times n$ zero-one matrix $A$ which avoids $P$. This maximum number is called the extremal function $ex(n,P)$.

Extremal theory of matrices have been well-studied. F\"{u}redi and Hajnal conjectured that $ex(n,P)=O(n)$ for all permutation matrices $P$ \cite{FH}. Klazar showed that this conjecture implies the Stanley-Wilf conjecture \cite{Kl}. Marcus and Tardos proved the  F\"{u}redi and Hajnal conjecture \cite{MT} and hence settled the Stanley-Wilf conjecture.  Keszegh conjectured that $ex(n, P)=O(n)$ for all tuple permutation matrices $P$ \cite{K}. Geneson proved that the conjecture is true \cite{G}.

Extremal theory of zero-one matrices have found applications to areas such as computational geometry and graph theory. For instance, extremal functions have been used to analyze the complexity of an algorithm for computing a shortest rectilinear path aviding rectilinear obstacles in the plane \cite{M}. Furthermore, if we associate two dimensional $0$-$1$ matrices with ordered bipartite graphs by relating rows and columns to the two ordered partite sets of vertices and  interpreting ones as edges, then this extremal problem can be viewed as the Tur\'an extremal problem for ordered bipartite graphs \cite{PT}.  When $R^{k,\ell}$ is a $k \times \ell$ matrix of all ones, the extremal problem $ex(n, R^{k,\ell})$ is the matrix version of the classical Zarankiewicz problem.  K\H{o}v\'{a}ri, S\'{o}s, and Tur\'{a}n found an upper bound $O(n^{2 - {\max(k, \ell) \over k \ell}})$ on $ex(n,R^{k,\ell})$ \cite{KST}. A lower bound $\Omega(n^{2 - {k + \ell - 2 \over k \ell -1}})$ was also known \cite{ES}.

For bounding extremal functions $ex(n, P)$ of forbidden $0$-$1$ matrices $P$, it can be useful to calculate exact values of the extremal function for small values of $n$. One way to do this is to check whether any of the $n \times n$ matrices with $k$ ones avoid $P$ for increasing values of $k$. To determine whether a $n\times n$ zero-one matrix $A$ contains a $k \times \ell$ zero-one matrix $P$, the naive algorithm would be to check every $k \times \ell$ submatrix of $A$ to see if any of them can be changed to $P$ by changing some ones to zeroes. This algorithm takes $O({n \choose k}{n \choose \ell})=O(n^{k+\ell})$ time.

For specific patterns $P$, we can come up with faster algorithms to see if an $n \times n$ matrix $A$ contains $P$. We define a zero-one rectangular matrix to be an $L$-shaped pattern if its first column and last row are both full of ones and it has zeroes elsewhere. More generally, we call a zero-one rectangular matrix a cross pattern if it has one row and one column both full of ones and zeroes elsewhere. We are also interesting in studying identity matrices, which are square matrices with ones on the diagonal and zeroes everywhere else, an example of permutation matrices in the F\"{u}redi and Hajnal conjecture  \cite{FH}. A $j$-tuple identity matrix is obtained by replacing each $1$ in an identity matrix with a column of $j$ ones and each $0$ in the identity matrix with a column of $j$ zeroes. Tuple identity matrices are a special case of the tuple permutation matrix studied by Geneson \cite{G}.

When $P$ is a column of all ones $R^{k,1}$, an identity matrix, a tuple identity matrix, an $L$-shaped pattern, or a cross pattern, we present algorithms to determine whether an $n\times n$ $0$-$1$ matrix $A$ contains $P$ with worst case running time of $O(n^2)$. This is significant, because the containment algorithm for any pattern $P$ must runs in worst case $\Omega(n^2)$, as we show in this paper. Therefore, the worst case running time for these patterns is $\Theta(n^2)$, the lowest possible order a containment algorithm can achieve.

When $P$ is a $k \times \ell$ all-ones matrix $R^{k,\ell}$, we present an algorithm that runs in time $O(n^{\min(k,\ell)+1})$, which is still an improvement on the naive algorithm, which runs in time $O(n^{k+\ell})$.

Our algorithms that we present are signficant because running these efficient algorithms may help us obtain data to improve known bounds on the matrix extremal function. For instance, if we run our algorithm for containment of $R^{k,\ell}$ we should be able to obtain exact values of $ex(n, R^{k,\ell})$ for small values of $n$. This test data may give us insight on how to narrow the gap between the known upper and lower bounds of $O(n^{2 - {\max(k, \ell) \over k \ell} })$ and $\Omega(n^{2 - {k + \ell - 2 \over k \ell -1}})$ on $ex(n, R^{k,\ell})$.

Our paper is organized as follows. In section 2, we use the extremal function to obtain preliminary bounds on containment algorithm complexity. In particular, we establish the worst case running time of $\Omega(n^2)$ of any containment algorithm. In section 3, we present our containment algorithms for several specific patterns $P$ and analyze their complexity. We present $\Theta(n^2)$ algorithms for a column of all ones, an identity matrix, a tuple identity matrix, an $L$-shaped pattern, and a cross pattern and a higher order algorithm for a rectangular all-ones matrix.

\medskip

\section{Containment algorithm complexity in terms of extremal functions}

In this section, we use the extremal function $ex(n,P)$ to obtain upper and lower bounds on the complexity of an algorithm to determine containment of $P$. Our first result shows that any containment algorithm has at least quadratic running time.

\begin{theo}
\label{lower}
For any $n \times n$ matrix $A$ and any pattern $P$, an algorithm to determine whether $A$ contains $P$ has worst case $\Omega(n^2)$.
\end{theo}
\begin{proof}
Any containment algorithm has to check at least $n^2-ex(n,P)$ entries to declare that $A$ avoids $P$. Otherwise, if the algorithms skip $ex(n,P)+1$ elements, then in the worst case these skipped entries are all $1$-entries and form a submatrix containing $P$. If $P$ is a $a \times b$ matrix, then we have $ex(n,P)\le ex(n,R^{a,b}) = O(n^{2-{\max(k,\ell) \over k\ell}}) = o(n^2)$ by \cite{KST} and since $R^{a,b}$ contains $P$. Therefore, it follows that the proportion of entries of $A$ that must be checked before concluding that $A$ avoids $P$ is close to $1$. This means that the running time to determine whether $A$ contains $P$ is $\Omega(n^2)$ for any pattern $P$.
\end{proof}

Our next result bounds the complexity of the containment algorithm in terms of the extremal function and the complexity in terms of the number of $1$-entries.

\begin{theo}
\label{upper}
Let $f$ be an increasing function. If $A$ is an $n \times n$ matrix with $x$ one-entries such that there is an $O(f(x))$ algorithm that decides if the given $x$ ones contain $P$ then there is an $O(n^2+f(ex(n,P)))$ algorithm that decides if $A$ contain $P$.
\end{theo}
\begin{proof}
First count the number of $1$-entries in $A$, and if there are more than $ex(n,P)$ ones in $A$, then $A$ contains $P$. It takes $O(n^2)$ time to count the $1$-entries of $A$. Otherwise, $A$ has at most $ex(n,P)$ ones. Now we can determine whether these at most $ex(n,P)$ ones contain $P$ in $O(f(ex(n,P))))$ time. Thus the total running time of our algorithm is $O(n^2+f(ex(n,P)))$.
\end{proof}

Our theorem above establishes the existence of a quadratic containment algorithm for certain patterns with linear extremal function.

\begin{cor}
\label{linear}
If $ex(n,P)=O(n)$ and there is an $O(x^2)$ algorithm to determine whether the $x$ one-entries in an $n\times n$ matrix $A$ contain $P$, then there is an algorithm running in time $\Theta(n^2)$ to determine whether $A$ contains $P$.
\end{cor}
\begin{proof}
Theorem \ref{upper} shows that there is an algorithm running in time $O(n^2+n^2)=O(n^2)$ and Theorem \ref{lower} shows that this algorithm must run in time $\Omega(n^2)$ .
\end{proof}

We can also establish a weaker bound of $o(n^4)$ for more general patterns $P$.

\begin{cor}
For any pattern $P$, if there is an $O(x^2)$ algorithm to determine whether the $x$ one-entries in an $n \times n$ matrix $A$ contain $P$, then there is an $o(n^4)$ algorithm to determine whether $A$ contains $P$.
\end{cor}
\begin{proof}
If $P$ is of size $k \times \ell$, then $R^{k,\ell}$ contains $P$ so that $ex(n,P) \le ex(n, R^{k,\ell})= O(n^{2-{\max(k,\ell) \over k\ell}}) = o(n^2)$ from \cite{KST}.
\end{proof}

In the next section, we present algorithms for many specific patterns $P$ with running time of $O(n^2)$. It thereby follows from Theorem \ref{lower} that their worst case running time is precisely $\Theta(n^2)$.

\medskip

\section{Algorithms for specific patterns $P$}

In this section, for special $k \times \ell$ matrices $P$ we improve the $O(n^{k+\ell})$ running time in the naive containment algorithm. For $P$ a column of all-ones, an identity matrix, a tuple identity matrix, an $L$-shaped pattern, or a cross pattern, we present algorithms that determine whether a $n\times n$ zero-one matrix $A$ contains $P$ in $O(n^2)$ time. It follows from Theorem \ref{lower} that their running time is $\Theta(n^2)$, the lowest possible order for a containment algorithm. When $P$ is a $k \times \ell$ matrix, our best algorithm runs in time $O(n^{\min(k,\ell)+1})$, which is still an improvement on the naive $O(n^{k+\ell})$ running time.

The easiest case is when $P$ is a column of ones.

\begin{theo} 
If $P$ is a $k \times 1$ all ones matrix, then there is an $O(n^2)$ algorithm to determine whether an arbitrary $n \times n$ zero-one matrix $A$ contains $P$.
\end{theo}
\begin{proof}
For a $n \times n$ zero-one matrix $A$ there is an $O(n^2)$ algorithm which scans $A$ column by column, and whenever it finds any column of $A$ with at least $k$ ones, it stops and determines $A$ contains $P$. Otherwise $A$ doesn't contain $P$ after the algorithm scans through all the columns.
\end{proof}

We present a more complicated algorithm to determine containment of an identity matrix.

\begin{theo} 
If $P$ is an identity matrix, then there is an $O(n^2)$ algorithm to determine whether an arbitrary $n \times n$ zero-one matrix $A$ contains $P$.
\end{theo}
\begin{proof}
Let $P$ be a $k\times k$ identity matrix. We maintain an array $D[n+1][n+1]$ and we follow matrix notation by starting indices from $1$ rather than from $0$. We set $D[r][c]=0$ if $r=n+1$ or $c=n+1$. For $1\leq r\leq n, 1 \leq c \leq n$, $D[r][c]$ is the maximum number of ones from column $c$ to the last column of $A$ such that (1) all these ones are in rows between $r$ and $n$, inclusive and (2) all these ones form a $D[r][c]\times D[r][c]$ identity matrix when we remove all columns and rows not containing these ones.

Initially $D[r][c]$ is $0$ for all $r$. The algorithm updates $D$ as follows.

\begin{lstlisting}
 For c = n to 1
   For r = n to 1
     D[r][c]=max(D[r+1][c],D[r][c+1])
     if A(r,c) = 1
       D[r][c]=max(D[r][c], 1 + D[r+1][c+1])
\end{lstlisting}

The algorithm reports $A$ contains $B$ whenever some $D[r][c]$ hits $k$. It is easy to see that this algorithm has complexity $O(n^2)$.
\end{proof}

Now we generalize our algorithm for the identity matrix to also work for the tuple identity matrix.

\begin{theo} 
If $P$ is a tuple identity matrix, then there is an $O(n^2)$ algorithm to determine whether an arbitrary $n \times n$ zero-one matrix $A$ contains $P$.
\end{theo}
\begin{proof}
Let $P$ be a $jk \times k$ tuple identity matrix, which is obtained by replacing each one of a $k \times k$ identity matrix with a $j \times 1$ all ones matrix and each zero of the identity matrix with a $j \times 1$ all zeroes matrix. The following algorithm, similar to the algorithm for the identity matrix, determines whether $A$ contains $P$ in time $O(n^2)$. 

Now $D[r][c]$ indicates the maximum width of a $j$ tuple identity matrix contained by the submatrix of $A$ within row $r$ to $n$ and column $c$ to $n$. The algorithm reports true if any $D[r][c]$ reaches $k$.

The algorithm proceeds as before, but for each column $c$ and each row $r$, if it exists, we need to know a row index $H(r,c)$, which is the smallest row index such that there are $j$ ones between rows $r$ and $H(r,c)-1$, inclusive, of column $c$.  In our outer loop of the algorithm which scans $A$ by column,  we use an overhead array called oneIndices, which is a list of the row indices of ones in the current column, so that we can compute $H(r,c)$ with complexity $O(n)$ per column. Therefore, this overhead keeps overall asymptotic complexity at $O(n^2)$. Our algorithm is written out in full below.

\begin{lstlisting}
 For c = n to 1
   oneIndices=[]
   For r = n to 1
     D[r][c]=max(D[r+1][c],D[r][c+1])
     if A(r,c) = 1
      oneIndices.append(r)
      if(len(oneIndices)>=j) 
        H(r,c)=oneIndices(len(oneIndices)-j)}
        D[r][c]=max(D[r][c], 1 + D[H(r,c)][c+1])
\end{lstlisting}
\end{proof}

Now we present an algorithm for containment of $L$-shaped patterns that also has complexity $O(n^2)$.

\begin{theo} 
If $P$ is an $L$-shaped pattern matrix, then there is an $O(n^2)$ algorithm to determine whether an arbitrary $n \times n$ zero-one matrix $A$ contains $P$.
\end{theo}
\begin{proof}
Let $P$ be a $m \times n$ matrix that is an $L$-shaped pattern of ones with $P_{i,j}=1$ iff $j=1$ or $i=m$.

Keep $n$ counters, one for each column. Scan $A$ row by row from bottom to top. If a row has $k \ge n$ ones, at column indices $x_1, \ldots, x_k$, then increment each of the counters $x_1, \ldots, x_{k-n+1}$ by $1$. Also for any of the column counters $x_{k-n+2}, \ldots, x_k$ which are already positive, increment them by one as well. Otherwise if a row has $k<n$ ones, increment any of the counters $x_1, \ldots, x_k$ which are already positive by one. Whenever a counter hits $m$, $A$ contains $P$. It is easy to see that this algorithm has complexity $O(n^2)$.
\end{proof}

Now we generalize our result for $L$-shaped patterns to cross patterns. However, our algorithm for general cross patterns is  more complex than our algorithm above for $L$-shaped patterns.

\begin{theo} 
If $P$ is a cross pattern matrix, then there is an $O(n^2)$ algorithm to determine whether an arbitrary $n \times n$ zero-one matrix $A$ contains $P$.
\end{theo}
\begin{proof}
Let $P$ be an $a \times b$ matrix that is a cross pattern with $P(i,j)=1$ iff $i=c$ or $j=d$ where $c$ and $d$ are constants such that $1 \le c \le a$ and $1 \le d \le b$.

Let $x$ be the number of $1$-entries in $A$. We show that there is an $O(x)$ algorithm to determine whether $A$ contains $P$. Since $x \le n^2$, this algorithm runs in time $O(n^2)$. Associate each $1$-entry $e_k$ in $A$ with at most $4$ links, or pointers, to other $1$-entries in $A$. Specifically, there is a link to the next $1$-entry to the right of $e_k$ in the same row, or if it doesn't exist, the left most $1$-entry in the next row below. There is another link to the next $1$-entry to the left of $e_k$ in the same row, or if it doesn't exist, the right most $1$-entry in the next row below. And the other two links are the two analogous vertical links. Furthermore, each $1$-entry may have special marks indicating that it is the top/bottom/right-most/left-most $1$-entry of that column/row.

Following these $4$ links and marks, we can compute and store $R_k$, the number of $1$-entries to the right of $e_k$ in the same row, and similary $L_k$, $U_k$, and $D_k$ in $4$ linear traversals, i.e. $O(x)$ time. Finally, in one extra traversal the algorithm reports that $A$ contains $P_{a \times b}$ if there exists an $1$-entry $e_i$ such that $L_i \ge d-1$, $R_i \ge b-d$, $U_i \ge c-1$, and $D_i \ge a-c$. Thus the algorithm to determine containment for $P_{a \times b}$ runs in $O(x)=O(n^2)$ time.
\end{proof}

Finally, we present a containment algorithm for a rectangular matrix of all ones. This algorithm does not run in $O(n^2)$ time, however.

\begin{theo} 
If $P$ is a $k \times \ell$ matrix of all ones, then there is an $O(n^{\min(k,\ell)+1})$ algorithm to determine whether an arbitrary $n \times n$ zero-one matrix $A$ contains $P$.
\end{theo}
\begin{proof}
An algorithm with complexity $O(n^3)$ and memory $O(n^2)$ can decide whether a given matrix $A$ contains $P$ if $P$ is a $k \times 2$ all-ones matrix. It scans $A$ row by row and keeps ${n \choose 2}$ counters for each unordered pair $(a,b)$ of columns $a$ and $b$. If both $A_{i,a}$ and $A_{i,b}$ in row $i$ are $1$ then counter $(a,b)$ is incremented by $1$. The algorithm reports true if any counter hits $k$. If $P$ is a $k \times \ell$ all-ones matrix, a similar algorithm determines containment with time complexity $O(n^{\ell+1})$. Similarly if we instead scan $A$ column by column first and use counters for pairs of rows, then we get an algorithm wth time complexity $O(n^{k+1})$. Taking the more efficient of these two algorithms gives our result.
\end{proof}

\section{Conclusion and Open Problems}
In this paper, we analyzed the complexity of algorithms which determine whether a given $n \times n$ matrix contains a specific pattern $P$. We gave $\Theta(n^2)$ algorithms for basic patterns $P$ such as identity matrices, tuple identity matrices, column all ones matrices, $L$-shaped matrices, and cross-patterns. We also obtained an $O(n^{\min(k,\ell)+1})$ algorithm when $P$ is a $k \times \ell$ all-ones matrix. 

For which patterns $P$ do we have a containment algorithm running in time $\Theta(n^2)$? To answer this question, it may be useful to rephrase the question in terms of the number of $1$-entries in $A$. Given a $0$-$1$ matrix $A$ with the $1$-entries $e_{1}, \ldots, e_{x}$, when can we determine whether $A$ contains $P$ with an $O(x)$ algorithm? This will guarantee an $O(n^2)$ containment algorithm. For any $0$-$1$ matrix $A$ with the $1$-entries $e_{1}, ..., e_{x}$, can we can always determine whether $A$ contains $P$ with an $O(x^2)$ algorithm? If this is true, then for patterns $P$ with $ex(n,P)=O(n)$ we would have an $O(n^2)$ containment algorithm.

We know that our $O(n^2)$ algorithms have the most efficient worst-case running time for any containment algorithm. However, it is not clear whether our $O(n^{\min(k,\ell)+1})$ algorithm for $P$ an all-ones $k \times \ell$ matrix can be improved. Our only known lower bound on an algorithm for this pattern $P$ is $\Omega(n^2)$. Therefore, we ask: what is the fastest running time for an algorithm to determine whether a $n \times n$ matrix $A$ contains a $k \times \ell$ matrix $B$ with all ones?

\section{Acknowledgments}
CrowdMath is an open program created by the MIT Program for Research in Math, Engineering, and Science (PRIMES) and Art of Problem Solving that gives high school and college students all over the world the opportunity to collaborate on a research project. The 2016 CrowdMath project is online at http://www.artofproblemsolving.com/polymath/mitprimes2016. 

\bibliographystyle{amsplain}

\end{document}